\documentclass[letterpaper, 12 pt, conference]{ieeeconf}  % Comment this line out
\IEEEoverridecommandlockouts                              % This command is only
\usepackage{graphicx} % for pdf, bitmapped graphics files
\usepackage{mathptmx} % assumes new font selection scheme installed
\usepackage{times} % assumes new font selection scheme installed
\usepackage{amsmath} % assumes amsmath package installed
\usepackage{amssymb}  % assumes amsmath package installed
\newtheorem{lem}{Lemma}
\newtheorem{thm}{Theorem}

\def\<{\leqslant}           % nice less than or equal to sign
\def\>{\geqslant}           % nice larger than or equal to sign
\def\d{\partial}

\def\wt{\widetilde}

\def\Re{{\rm Re}}   % real part
\def\Im{{\rm Im}}   % imaginary part

\def\col{{\rm vec}}   % vectorization of matrices
\def\cH{{\cal H}}   % Hardy space
\def\mA{{\mathbb A}}    % space of real antisymmetric matrices
\def\mR{{\mathbb R}}    % real line
\def\mC{{\mathbb C}}    % complex plane

\def\Tr{{\rm Tr}}       % matrix trace
\def\tr{{\rm tr}}       % matrix trace
\def\rT{{\rm T}}        % matrix transpose
\def\diam{\diamond}       % matrix trace
\def\bE{{\bf E}}    % expectation

\def\[[[{[\![\![}
\def\]]]{]\!]\!]}

\def\bra{{\langle}}
\def\ket{{\rangle}}

\def\Bra{\big\langle}
\def\Ket{\big\rangle}

\def\re{{\rm e}}        % number e
\def\rd{{\rm d}}        % differential

\def\fO{{\mathfrak O}}

\def\cL{{\mathcal L}}

\def\bJ{{\bf J}}

\def\br{{\bf r}}
\def\x{\times}
\def\ox{\otimes}

\def\cZ{{\mathcal Z}}

\def\bH{{\mathbf H}}

\def\cF{{\cal F}}
\def\cW{{\mathcal W}}

\def\cX{{\mathcal X}}

\def\cM{{\cal M}}
\def\cV{{\cal V}}
\def\cC{{\cal C}}
\def\cR{{\cal R}}

\def\sB{{\sf B}}
\def\sC{{\sf C}}
\def\sE{{\sf E}}

\def\cA{{\cal A}}
\def\cB{{\cal B}}
\def\cE{{\cal E}}
\def\cov{{\bf cov}}

\def\mU{{\mathbb U}}

\def\mH{{\mathbb H}}
\def\mS{{\mathbb S}}

\def\veps{\varepsilon}

\def\ups{\upsilon}
\def\Ups{\Upsilon}

\title{\Large \bf
Characterization and
Moment Stability Analysis of \\ Quasilinear Quantum Stochastic Systems with \\ Quadratic Coupling to External Fields
}

\author{Igor G. Vladimirov, \qquad Ian R. Petersen%
\thanks{This work is supported by the Australian Research Council. The authors are with the School of Engineering and Information Technology, University of New South Wales at the Australian Defence Force Academy, Canberra, ACT 2600, Australia. E-mail: {\small\tt igor.g.vladimirov@gmail.com, i.r.petersen@gmail.com}.
}
}
\onecolumn
\begin{document}
\maketitle
\thispagestyle{empty}
%\pagestyle{empty}

%\maketitle
%18 August 2011

%\thispagestyle{empty}
%\pagestyle{empty}

%%%%%%%%%%%%%%%%%%%%%%%%%%%%%%%%%%%%%%%%%%%%%%%%%%%%%%%%%%%%%%%%%%%%%%%%%%%%%%%%%%%%%%%%%%%%%%%%%%%
\begin{abstract}
The paper is concerned with open quantum systems whose Heisenberg dynamics are described by quantum stochastic differential equations driven by external boson fields. The system-field coupling operators are assumed to be quadratic polynomials of the system observables, with the latter satisfying canonical commutation relations. In combination with a cubic system Hamiltonian, this leads to a class of quasilinear quantum stochastic systems which retain algebraic closedness in the evolution of mixed moments of the observables. Although such a system is nonlinear and its quantum state is no longer Gaussian,  the dynamics of the moments of any order are amenable to exact analysis, including  the computation of their steady-state values. In particular, a generalized criterion is developed for quadratic stability of the quasilinear systems.
The results of the paper are applicable to the generation of non-Gaussian quantum states with manageable moments and an optimal design of linear quantum  controllers for quasilinear quantum plants.
\end{abstract}
%%%%%%%%%%%%%%%%%%%%%%%%%%%%%%%%%%%%%%%%%%%%%%%%%%%%%%%%%%%%%%%%%%%%%%%%%%%%%%%%%%%%%%%%%%%%%%%%%%%

%\begin{keywords}
%quantum stochastic system, canonical commutation relation, coupling operator, quasilinear quantum system, algebraic moment closedness, non-Gaussian quantum state, quadratic stability.
%\end{keywords}

%%%%%%%%%%%%%%%%%%%%%%%%%%%%%%%%%%%%%%%%%%%%%%%%%%%%%%%%%%%%%%%%%%%%%%%%%%%%%%%%%%%%%%%%%%%%%%%%%%%
\section{INTRODUCTION}\label{sec:intro}
%%%%%%%%%%%%%%%%%%%%%%%%%%%%%%%%%%%%%%%%%%%%%%%%%%%%%%%%%%%%%%%%%%%%%%%%%%%%%%%%%%%%%%%%%%%%%%%%%%%

The paper is concerned with open quantum systems whose Heisenberg dynamics are described by quantum stochastic differential equations (QSDEs) driven by external boson fields \cite{P_1992}. This class of dynamical systems,  whose variables are noncommutative operators on a Hilbert space, is a common source of models in quantum optics \cite{GZ_2004} and quantum control \cite{EB_2005,JG_2010,JNP_2008}. Central to this approach  is a Markovian description of interaction between the quantum-mechanical  system and its environment, whose ``energetics'' is specified by the system Hamiltonian and system-field coupling operators. The case where the Hamiltonian is quadratic and the coupling operators are linear  in the system observables, satisfying canonical commutation relations (CCRs) \cite{M_1998}, corresponds to an open quantum harmonic oscillator. Its importance for  linear quantum stochastic control \cite{EB_2005,NJP_2009,P_2010,VP_2011a} is explained by tractability of the dynamics of moments of the system observables, which is closely related to the invariance of the class of Gaussian quantum states \cite{KRP_2010} of the system subject to external fields in the vacuum state \cite{P_1992}. This allows the linear quantum control to inherit  at least some features of the classical control schemes  \cite{KS_1972}, including practical computability of quadratic costs.
In the present paper, we  consider more complicated system-field interactions, described by coupling operators which are quadratic polynomials of the system observables. In combination with a cubic system Hamiltonian, this leads to a novel class of \emph{quasilinear} quantum stochastic systems which extend the linear quantum dynamics above and yet retain algebraic closedness in the evolution of mixed moments of the observables, similarly to their classical counterparts in \cite{W_1966}. Although such a system is nonlinear and its quantum state is no longer Gaussian,  the dynamics of the moments of any order are amenable to exact analysis, which includes  the computation of their steady-state values. In particular, this allows a generalized criterion to be developed for quadratic stability of the quasilinear systems.
The results of the paper are applicable to the generation of non-Gaussian quantum states \cite{Y_2009,ZJ_2011b} with manageable moment dynamics. They can also be used for an optimal design of linear quantum  controllers for quasilinear quantum plants.

\section{OPEN QUANTUM SYSTEMS}\label{sec:class}
%%%%%%%%%%%%%%%%%%%%%%%%%%%%%%%%%%%%%%%%%%%%%%%%%%%%%%%%%%%%%%%%%%%%%%%%%%%%%%%%%%%%%%%%%%%%%%%%%%%

We will briefly review the underlying class of models for open quantum systems which is based on QSDEs.
Suppose  $W(t):= (W_k(t))_{1\< k \< m}$ is an $m$-dimensional  quantum Wiener process of time $t\> 0$ on a boson Fock space $\cF$ \cite{P_1992}, with the quantum Ito table
\begin{equation}
\label{Omega}
    \rd W\rd W^{\rT}
    :=
    (\rd W_j \rd W_k)_{1\< j,k\< m}
    =
    \Omega \rd t,
\end{equation}
where the time argument  of $W$ is omitted for the sake of brevity, and $\Omega:= (\omega_{jk})_{1\< j,k\< m}$ is a constant complex positive semi-definite Hermitian matrix of order $m$. Here, a complex number $\varphi$ is identified with the operator $\varphi I_{\cF}$, where $I_{\cF}$ denotes the identity operator on $\cF$. Also, vectors are  organised as columns unless indicated otherwise, and   the transpose $(\cdot)^{\rT}$ acts on matrices with operator-valued entries as if the latter were scalars. The  entrywise real part of the matrix $\Omega$, denoted by
\begin{equation}
\label{V}
    V:= (v_{jk})_{1\< j,k\< m}:=\Re \Omega,
\end{equation}
is a positive semi-definite symmetric matrix.
The entries $W_1(t), \ldots, W_m(t)$ of the vector $W(t)$ are self-adjoint operators on $\cF$, associated with the annihilation and creation operator  processes of external boson fields.  The imaginary part of the quantum Ito matrix  $\Omega$  in (\ref{Omega}) describes the CCRs between the fields in the sense that
\begin{equation}
\label{CCRW}
    [\rd W, \rd W^{\rT}]
      :=
    ([\rd W_j, \rd W_k] )_{1\< j,k\< m}\\
     =  i J\rd t,
    \quad
    J:= 2 \Im \Omega,
\end{equation}
where $[A, B]:= AB-BA$ is the commutator of operators which applies entrywise, and $i:= \sqrt{-1}$ is the imaginary unit. We consider an open quantum system (further referred to as the \textit{plant}) whose state variables $X_1(0), \ldots, X_n(0)$ at time $t = 0$ are self-adjoint operators on a complex separable Hilbert space $\cH$. Any such operator $\xi$ can be identified with its ampliation $\xi\ox I_{\cF}$ on the tensor  product space $\cH \ox \cF$. Suppose the plant interacts only with the external fields,   so that the density operator $\rho(t)$,  which describes the quantum state \cite{M_1998}
of the plant-field (as a closed system, isolated from the environment) on $\cH\ox \cF$ at time $t$, evolves in the Schr\"{o}dinger picture of quantum dynamics as
\begin{equation}
\label{rhot}
    \rho(t)
    =
    U(t) \rho(0)  U(t)^{\dagger}.
\end{equation}
Here, $U(t)$ is a unitary operator on the Hilbert space $\cH\ox \cF$, initialized at the identity operator $U(0)= I_{\cH \ox \cF}$, and $(\cdot)^{\dagger}$ is the operator adjoint.
The initial quantum state of the plant-field composite  system is assumed to be the tensor product
\begin{equation}
\label{rho0}
    \rho(0)
    :=
    \varpi(0)
    \ox
    \ups
\end{equation}
of the initial plant state $\varpi(0)$ on $\cH$ and the pure state $\ups := |0\ket\bra 0|$ of the external field associated with the vacuum vector $|0\ket$ in $\cF$, where use has been made of the Dirac bra-ket notation \cite{M_1998}.   In the Heisenberg picture, an observable $\xi(t)$ on $\cH\ox \cF$ evolves in a dual unitary fashion to (\ref{rhot}):
\begin{equation}
\label{xit}
    \xi(t) =
    U(t)^{\dagger} \xi(0) U(t),
\end{equation}
with the duality being understood in the sense of the equivalence
\begin{equation}
\label{bE}
    \bE\xi(t)
    :=
    \Tr(\rho(0) \xi(t) )
    =
    \Tr(\rho(t) \xi(0))
\end{equation}
between two representations of the quantum expectation. The unitary operator $U(t)$ itself is driven by the internal dynamics of the plant (which the plant would have in isolation from the surroundings) and by the plant-field interaction. In the weak interaction limit, which neglects the influence of the plant on the Markov structure of the field in the vacuum state, a wide class of open quantum systems is captured by the following QSDE:
\begin{equation}
\label{dU}
    \rd U
    =
    -((iH+h^{\rT} \Omega h/2)\rd t +ih^{\rT}\rd W  ) U.
\end{equation}
Here,
$H$ is the plant Hamiltonian,
%, which quantifies the ``internal'' energetics of the plant,
and  $h:= (h_j)_{1\< j \< m} $ is a vector of plant-field coupling  operators, with $h_j$ pertaining to the interaction between the plant and the  $j$th external field. Both $H$ and $h_1, \ldots, h_m$ are self-adjoint operators on $\cH$, which are  usually functions of the plant observables $X_1(0), \ldots, X_n(0)$. The term $h^{\rT}\rd W = \sum_{k=1}^m h_k \rd W_k$ in (\ref{dU}), which can be interpreted as an incremental perturbation to the plant Hamiltonian $H$ due to the interaction with the external fields,   is an alternative form of the more traditional representation
\begin{equation}
\label{hL}
    h^{\rT}\rd W
    =
    i(L^{\rT}\rd \cA^{\#}- L^{\dagger} \rd \cA)
\end{equation}
through the $m/2$-dimensional field annihilation  and creation operator processes $\cA(t)$ and $\cA(t)^{\#}$ on the Fock space $\cF$ with the quantum Ito table
\begin{equation}
\label{dAA}
    \rd
    {\small\begin{bmatrix}
        \cA\\
        \cA^{\#}
    \end{bmatrix}}
    \rd
    {\small\begin{bmatrix}
        \cA^{\dagger} &
        \cA^{\rT}
    \end{bmatrix}}
    =
    {\small\begin{bmatrix}
        I_{m/2} & 0\\
        0 & 0
    \end{bmatrix}}
    \rd t,
\end{equation}
so that $[\rd \cA, \rd \cA^{\dagger}]:= \rd \cA \rd \cA^{\dagger}-\rd \cA^{\#}\rd \cA^{\rT} =   I_{m/2}\rd t$. Here, $m$ is assumed to be even, $(\cdot)^{\dagger} := ((\cdot)^{\#})^{\rT}$ denotes the transpose of the entrywise adjoint  $(\cdot)^{\#}$, and $I_r$ is the identity matrix of order $r$. In application to ordinary matrices, $(\cdot)^{\dagger}$ reduces to  the complex conjugate transpose $(\cdot)^* := (\overline{(\cdot)})^{\rT}$. The representation (\ref{hL}) corresponds to the case when the scattering matrix \cite{P_1992} is the identity matrix.
The vector $L := (L_k)_{1\< k\< m/2}$  consists of linear operators on $\cH$, which are not necessarily self-adjoint. The relation of $h$, $W$ with $L$, $\cA$ is described by
\begin{align}
\label{hW1}
    h
    & =
    \frac{1}{\sqrt{2}}
    \left({\small\begin{bmatrix}
        i & -i\\
        1 & 1
    \end{bmatrix}}
    \ox I_{m/2}\right)
    {\small \begin{bmatrix}
        L\\
        L^{\#}
    \end{bmatrix}},\\
\label{hW2}
    W
    & =
    \frac{1}{\sqrt{2}}
    \left({\small
    \begin{bmatrix}
        1 & 1\\
        -i & i
    \end{bmatrix}}
    \ox I_{m/2}
    \right)
    {\small\begin{bmatrix}
        \cA\\
        \cA^{\#}
    \end{bmatrix}},
\end{align}
where $\ox$ denotes the Kronecker product of matrices,
so that the corresponding quantum Ito matrix $\Omega$ in (\ref{Omega}) is $\Omega = \Big(I_m + i {\scriptsize\begin{bmatrix}0 & 1\\ -1 & 0\end{bmatrix}}\ox I_{m/2}\Big)\big/2$.
In general, the operator
$
    h^{\rT} \Omega h
    :=
    \sum_{j,k=1}^m
    \omega_{jk} h_j h_k
$,
which takes the form $h^{\rT} \Omega h = L^{\dagger}L = \sum_{k=1}^{m/2} L_k^{\dagger}L_k$ in the case (\ref{dAA})--(\ref{hW2}),
is self-adjoint since $(h^{\rT} \Omega h)^{\dagger} = \sum_{j,k=1}^m
\overline{\omega_{jk}} h_k h_j = h^{\rT} \Omega^* h = h^{\rT} \Omega h$. The formulation using $h$ and $W$, in principle,  allows $W$ to be an additive mixture of the quantum noise with a classical random component such as the standard Wiener process. In view of (\ref{CCRW}), the ``magnitude'' of $\Im \Omega=J/2$ in comparison with $\Re \Omega$ (measured, for example, by the spectral radius $\br(\Im \Omega (\Re \Omega)^{-1})$ which is well defined and is strictly less than one if $\Omega \succ 0$) indicates the relative amount of ``quantumness'' in $W$. This setting reduces to the classical noise situation if the matrix $\Omega$ in (\ref{Omega}) is real, in which case $U(t)$ becomes a random process with values among unitary operators on $\cH\ox \cF$; see \cite{K_1972} and \cite[pp. 258--260]{P_1992}. The general situation is treated
by applying the quantum Ito rule $\rd (\eta\zeta) = (\rd \eta) \zeta + \eta\rd \zeta + (\rd \eta)\rd \zeta$
%(which respects the order of multiplication of noncommuting operators)
and using (\ref{dU}) along with the unitarity of $U(t)$ and  commutativity between the forward increment $\rd W$ and the adapted processes. This yields the following QSDE for the density operator $\rho(t)$ in (\ref{rhot}):
\begin{equation}
\label{drhot}
    \rd \rho =
    -i([H,\rho]\rd t + [h^{\rT},\rho]\rd W)
    +
    \tr(\Omega^{\rT} C(\rho))\rd t,
\end{equation}
which is referred to as the stochastic quantum master equation  \cite{GZ_2004}  in the Schr\"{o}dinger picture. Here, use is made of a self-adjoint operator
$
    \tr(\Omega^{\rT} C(\rho))
    :=
    \sum_{j,k=1}^m
    \omega_{jk}
    C_{jk}(\rho)
$
on $\cH\ox \cF$, with $\tr(\cdot)$ denoting a ``symbolic'' trace (to be distinguished from the complex-valued trace $\Tr(\cdot)$ of an operator),
where the matrix $C(\rho):= (C_{jk}(\rho))_{1\< j,k\< m}$ has operator-valued entries
\begin{equation}
\label{Cjk}
    C_{jk}(\rho)
    :=
    h_k\rho h_j - (h_jh_k\rho + \rho h_jh_k)/2.
\end{equation}
Note that $C_{jk}(\rho)^{\dagger} = C_{kj}(\rho)$ for any self-adjoint operator $\rho$ in view of self-adjointness of $h_1, \ldots, h_m$, and this, together with $\Omega^* = \Omega$, ensures the self-adjointness of $\tr(\Omega^{\rT} C(\rho))$.

%%%%%%%%%%%%%%%%%%%%%%%%%%%%%%%%%%%%%%%%%%%%%%%%%%%%%%%%%%%%%%%%%%%%%%%%%%%%%%%%%%%%%%%%%%%%%%%%%%%
\section{DECOHERENCE OPERATOR}\label{sec:DDR}
%%%%%%%%%%%%%%%%%%%%%%%%%%%%%%%%%%%%%%%%%%%%%%%%%%%%%%%%%%%%%%%%%%%%%%%%%%%%%%%%%%%%%%%%%%%%%%%%%%%

In the Heisenberg picture, an observable $\xi(t)$ on the composite Hilbert space $\cH\ox \cF$, which undergoes  the evolution (\ref{xit}), satisfies the QSDE
\begin{equation}
\label{dxit}
    \rd \xi =
    i([H,\xi]\rd t + [h,\xi]^{\rT}\rd W)
    +
    \cL(\xi)\rd t.
\end{equation}
Here, both the plant Hamiltonian $H(t)$ and the vector $h(t)$ of plant-field coupling operators are also evolved  by the  flow (\ref{xit}). However,  they depend on the vector $X(t):= (X_k(t))_{1\< k \< n}$ of the plant observables in the same way as $H(0)$ and $h(0)$ do  on $X(0)$.  Also,  $\cL$ denotes the Gorini-Kossakowski-Sudarshan-Lind\-blad (GKSL) superoperator defined by
\begin{equation}
\label{cL}
    \cL(\xi)
    :=
    \tr(\Omega^{\rT} D(\xi))
    =
    \sum_{j,k=1}^m
    \omega_{jk}
    D_{jk}(\xi),
\end{equation}
where $D(\xi):= (D_{jk}(\xi))_{1\< j,k\< m}$ is a matrix with operator-valued entries
\begin{align}
\nonumber
    D_{jk}(\xi)
    & :=
    h_j\xi h_k - (h_jh_k\xi + \xi h_jh_k)/2\\
\label{Djk}
    & =
    (h_j[\xi,h_k] + [h_j,\xi]h_k)/2.
\end{align}
The superoperators $D_{jk}$ are dual  to $C_{jk}$ from (\ref{Cjk}) in the sense that $ \Tr(C_{jk}(\rho) \xi)=\Tr(\rho D_{jk}(\xi))$. The superoperator matrix $D$ acts on an observable $\xi$ as
\begin{equation}
\label{Dmat}
    D(\xi)
     =
    \big(h[\xi,h^{\rT}] + [h, \xi]h^{\rT}\big)/2.
\end{equation}
The superoperators $C_{jk}$ and their duals $D_{jk}$, defined by (\ref{Cjk}), (\ref{Djk}), play an important  role in  the generators of quantum dynamical semigroups \cite{%A_2000,
GKS_1976,L_1976}. One of such semigroups governs the evolution of the  reduced plant density operator $\varpi(t)$ obtained by ``tracing out'' the quantum noise in (\ref{drhot}) over the field vacuum state $\ups$ which yields an ODE
$    \dot{\varpi} =
    -i[H,\varpi]
    +
    \tr(\Omega^{\rT} C(\varpi))
$.
The generator of the corresponding semigroup in the dual Heisenberg picture is $i[H,\cdot] + \cL$, where the superoperator $\cL$,  given by (\ref{cL}), is responsible for decoherence \cite{GZ_2004}  understood
as the  deviation from a unitary evolution which the plant observables would have alone in the absence of interaction with the environment. It is convenient to apply the QSDE (\ref{dxit}) entrywise to the vector $X(t)$ of plant observables as
\begin{equation}
\label{XQSDEgen}
    \rd X =
    F\rd t + G\rd W.
\end{equation}
The $n$-dimensional \emph{drift vector} $F(t)$ and the \emph{dispersion} $(n\x m)$-matrix $G(t)$ of this QSDE, defined by
\begin{equation}
\label{FG}
    F
    :=
    i[H,X] +\cL(X),
    \qquad
        G
    :=
    - i[X,h^{\rT}],
\end{equation}
are completely specified by the plant Hamiltonian $H$, the quantum Ito matrix  $\Omega$ of the field process $W$ from (\ref{Omega}), and the vector $h$ of plant-field coupling operators.

%%%%%%%%%%%%%%%%%%%%%%%%%%%%%%%%%%%%%%%%%%%%%%%%%%%%%%%%%%%%%%%%%%%%%%%%%%%%%%%%%%%%%%%%%%%%%%%%%%%
\begin{lem}
\label{lem:DDR}
The GKSL superoperator (\ref{cL}), applied to the vector $X$ of plant observables,  can be computed in terms of the dispersion matrix $G$ from (\ref{FG}) as
\begin{equation}
\label{DDR}
    \cL(X)
    =
    \frac{1}{2}
    \Big(
        GJh + i    \sum_{j,k=1}^m
    \omega_{jk}
    [h_j,g_k]
    \Big).
\end{equation}
Here, the matrix $J$ is defined by (\ref{CCRW}), and $g_1, \ldots, g_m$ denote the columns of $G$:
\begin{equation}
\label{gk}
    g_k = i[h_k, X].
\end{equation}
\end{lem}
%%%%%%%%%%%%%%%%%%%%%%%%%%%%%%%%%%%%%%%%%%%%%%%%%%%%%%%%%%%%%%%%%%%%%%%%%%%%%%%%%%%%%%%%%%%%%%%%%%%

\begin{proof} In view of the antisymmetry of the commutator,  it follows  from (\ref{FG}) that
\begin{equation}
\label{DDR1}
    \sum_{j,k=1}^m
    \omega_{jk}
        [h_j, X]h_k
         =
        -[X,h^{\rT}] \Omega h
        =
        -i G \Omega h.
\end{equation}
Therefore, since the quantum Ito matrix $\Omega$ in (\ref{Omega}) is Hermitian, then
\begin{equation}
\label{DDR2}
    \!\!\Big(\!
    \sum_{j,k=1}^m\!\!\!
    \omega_{jk} h_j[X,h_k]
    \Big)^{\#}\!\!
     = \!\!
        \sum_{j,k=1}^m\!\!
    \overline{\omega_{jk}} [h_k,X]h_j\!\!
     = \!\!
        \sum_{j,k=1}^m\!\!
    \omega_{jk} [h_j,X]h_k.\!\!\!\!\!\!
\end{equation}
By combining (\ref{DDR1}) with (\ref{DDR2}), it follows from (\ref{cL}), (\ref{Djk}) that
\begin{equation}
\label{DDR3}
    \cL(X)
    =
    i
    \big(
        (G \Omega h)^{\#}
        -
        G \Omega h
    \big)/2.
\end{equation}
In terms of the columns of the dispersion matrix $G$ in (\ref{gk}),
\begin{align}
\nonumber
    (G\Omega h)^{\#}
     = &
    \Big(
    \sum_{j,k=1}^m
    \omega_{jk}  g_j h_k
    \Big)^{\#}=
    \sum_{j,k=1}^m
    \overline{\omega_{jk}}
    (g_j h_k -[g_j,h_k])\\
\label{DDR4}
     = &
    G\overline{\Omega} h
    +
    \sum_{j,k=1}^m
    \omega_{jk}
    [h_j,g_k].
\end{align}
By substituting (\ref{DDR4}) into (\ref{DDR3}) and using the relationship $i(\overline{\Omega}-\Omega) = J$ from (\ref{CCRW}), it follows that
$    \cL(X)
    =
    i\Big(G (\overline{\Omega}-\Omega) h +
    \sum_{j,k=1}^m
    \omega_{jk}
    [h_j,g_k]
\Big)\big/2
     =
    GJh/2 + i\sum_{j,k=1}^m
    \omega_{jk}
    [h_j,g_k]/2
$,
which establishes (\ref{DDR}).\end{proof}
%%%%%%%%%%%%%%%%%%%%%%%%%%%%%%%%%%%%%%%%%%%%%%%%%%%%%%%%%%%%%%%%%%%%%%%%%%%%%%%%%%%%%%%%%%%%%%%%%%%

%Lemma~\ref{lem:DDR} reveals its utility in the case where the plant-field coupling operators $h_1, \ldots, h_m$ are polynomials of the plant observables and the latter satisfy CCRs.
In the next section, we will employ Lemma~\ref{lem:DDR} in order to review the computation of the drift vector  $F$ and the dispersion matrix $G$ for a class of linear open quantum systems.

%%%%%%%%%%%%%%%%%%%%%%%%%%%%%%%%%%%%%%%%%%%%%%%%%%%%%%%%%%%%%%%%%%%%%%%%%%%%%%%%%%%%%%%%%%%%%%%%%%%
\section{LINEAR PLANT-FIELD COUPLING}\label{sec:lincoup}
%%%%%%%%%%%%%%%%%%%%%%%%%%%%%%%%%%%%%%%%%%%%%%%%%%%%%%%%%%%%%%%%%%%%%%%%%%%%%%%%%%%%%%%%%%%%%%%%%%%

Omitting the time dependence, suppose the plant observables $X_1, \ldots, X_n$, which are assembled into the  vector $X$,  satisfy CCRs
\begin{equation}
\label{Theta}
    [X, X^{\rT}]
    :=
    \big(
        [X_j,X_k]
    \big)_{1\< j,k\< n}
    =
    i
    \Theta,
\end{equation}
where $\Theta:= (\theta_{jk})_{1\<j, k\< n}$ is a constant real antisymmetric matrix of order $n$ (we denote the space of such matrices by $\mA_n$). %Here, in a similar way to (\ref{CCRW}), a real $\varphi$ is identified with the observable $\varphi I$, where $I$ denotes the identity operator on $\cH$.
In this case, if the plant-field coupling operators $h_1, \ldots, h_m$ are polynomials of degree $r$ in the plant observables, then  the entries of the dispersion matrix $G$ in (\ref{FG}) and the vector  $\cL(X)$ in (\ref{DDR}) are polynomials of degrees $r-1$ and $2r-1$, respectively. This property follows from the reduction of a polynomial degree under taking the commutator with the observables due to the CCRs (\ref{Theta}):
\begin{equation}
\label{polred}
    [
        \Xi_k,
        X_{\ell}
    ]
    =
    i
    \sum_{j=1}^r
    \theta_{k_j \ell}
    \Xi_{k_1\ldots k_{j-1} k_{j+1}\ldots k_r},
\end{equation}
where
\begin{equation}
\label{Xi}
            \Xi_k
    :=
    X_{k_1} \x \ldots \x X_{k_r}
\end{equation}
denotes a degree $r$ monomial of the plant observables specified by an $r$-index $k:= (k_1, \ldots, k_r) \in \{1,\ldots, n\}^r$, with the order of multiplication being essential in the noncommutative case. The right-hand side of (\ref{polred}) is a polynomial of degree $r-1$.
In the plant-field interaction model which  is used in linear quantum control \cite{EB_2005,JNP_2008,NJP_2009,P_2010},  the vector $h$ of coupling operators depends linearly on $X$ in the sense that
\begin{equation}
\label{hM}
    h:= MX
\end{equation}
for some matrix $M\in \mR^{m\x n}$. In this case, the dispersion matrix $G$ in (\ref{FG}) becomes a constant real matrix,  since
\begin{equation}
\label{BM}
    G
    =
    -i[X,h^{\rT}]
    =
    -i[X, X^{\rT}]M^{\rT}
    =
    \Theta M^{\rT},
\end{equation}
where the bilinearity of the commutator is combined with the CCRs (\ref{Theta}). In view of Lemma~\ref{lem:DDR},  this allows $\cL(X)$ to inherit from $h$  the linearity with respect to the plant observables.

%%%%%%%%%%%%%%%%%%%%%%%%%%%%%%%%%%%%%%%%%%%%%%%%%%%%%%%%%%%%%%%%%%%%%%%%%%%%%%%%%%%%%%%%%%%%%%%%%%%
\begin{lem}
In the case of CCRs (\ref{Theta})  and linear plant-field coupling (\ref{hM}), the vector $X$ of plant observables satisfies a QSDE
\begin{equation}
\label{XQSDE}
    \rd X =
    (i[H,X]+KX)\rd t + B\rd W,
\end{equation}
where the matrices $K \in\mR^{n\x n}$ and $B \in \mR^{n\x m}$ are related to the the matrix $J$ from (\ref{CCRW}) by
\begin{equation}
\label{KB}
    \qquad
    K:= BJM/2,
    \qquad
    B:= \Theta M^{\rT}.
\end{equation}
%is a skew-Hamiltonian matrix with respect to the symplectic structure matrix $\Theta$.
\end{lem}
%%%%%%%%%%%%%%%%%%%%%%%%%%%%%%%%%%%%%%%%%%%%%%%%%%%%%%%%%%%%%%%%%%%%%%%%%%%%%%%%%%%%%%%%%%%%%%%%%%%

\begin{proof}
The constancy of the dispersion matrix $G$, computed  in (\ref{BM}), implies that the commutators on the right-hand side of (\ref{DDR}) vanish, that is, $[h_j, g_k] = 0$ for all $1\< j,k\< m$. Therefore,
\begin{equation}
\label{cLlin}
    \cL(X)
    =
        GJh/2
        =
        \Theta M^{\rT}J MX/2.
\end{equation}
The QSDE (\ref{XQSDE}) can now obtained from (\ref{XQSDEgen}) by substituting (\ref{BM}) and (\ref{cLlin}) into (\ref{FG}) and using the notation (\ref{KB}). \end{proof}
%%%%%%%%%%%%%%%%%%%%%%%%%%%%%%%%%%%%%%%%%%%%%%%%%%%%%%%%%%%%%%%%%%%%%%%%%%%%%%%%%%%%%%%%%%%%%%%%%%%

Since $K$ and $B$ are constant matrices, the QSDE (\ref{XQSDE})  may acquire nonlinearity with respect to the plant observables only through a nonquadratic  part of the plant Hamiltonian $H$. Indeed, if $H$ is a quadratic polynomial, that is,
\begin{equation}
\label{Hquad}
    H
    :=
    \sum_{k=1}^n
    \Big(
        \gamma_k
        +
        \frac{1}{2}
        \sum_{j=1}^n
        r_{jk}X_j
    \Big)
    X_k
    =
    (\gamma+RX/2)^{\rT} X,
\end{equation}
where $\gamma:= (\gamma_k)_{1\< k\< n} \in \mR^n$ is a given real vector, and $R:= (r_{jk})_{1\<j,k\< n}$ is a given real symmetric matrix of order $n$ (we denote the space of such matrices by $\mS_n$), then the commutator identities \cite[Eq. (3.50) on p. 38]{M_1998}) and the CCRs (\ref{Theta}) imply that
\begin{align}
\nonumber
    i[H,X]
     =&
    -i
    \sum_{k=1}^n
    \Big(
        \gamma_k [X, X_k]
        +
        \frac{1}{2}
        \sum_{j=1}^n
        r_{jk}
        [X, X_j X_k]
    \Big)\\
\nonumber
     = &
        \Theta \gamma
    -\frac{i}{2}
    \sum_{j,k=1}^n
    r_{jk}
    \big(
        [X,X_j]X_k
        +
        X_j[X,X_k]
    \big)\\
\label{HXquad}
     = &
        \Theta \gamma
    +\frac{1}{2}
    \sum_{j,k=1}^n
    r_{jk}
    (\Theta_j X_k + \Theta_k X_j)
     =
    \Theta (\gamma + R X),
\end{align}
where $\Theta_{\ell } := (\theta_{k\ell})_{1\< k\< n}$ denotes the $\ell$th column of the CCR matrix $\Theta$. Since the right-hand side of (\ref{HXquad})
depends affinely on $X$, then
(\ref{XQSDE}) becomes linear with respect to the plant observables:
\begin{equation}
\label{XQSDElin}
    \rd X =
    (AX + \Theta \gamma)\rd t + B\rd W,
    \qquad
    A := \Theta R + K,
\end{equation}
which corresponds to an open  quantum harmonic oscillator \cite{EB_2005}, that is, a common model employed in linear quantum control.

%%%%%%%%%%%%%%%%%%%%%%%%%%%%%%%%%%%%%%%%%%%%%%%%%%%%%%%%%%%%%%%%%%%%%%%%%%%%%%%%%%%%%%%%%%%%%%%%%%%
\section{ALGEBRAIC CLOSEDNESS IN MOMENT DYNAMICS}\label{sec:mom}
%%%%%%%%%%%%%%%%%%%%%%%%%%%%%%%%%%%%%%%%%%%%%%%%%%%%%%%%%%%%%%%%%%%%%%%%%%%%%%%%%%%%%%%%%%%%%%%%%%%

The linearity of the QSDE (\ref{XQSDElin}) ensures \emph{algebraic closedness} in the evolution of the mixed moments  of the plant observables, defined by
\begin{equation}
\label{muXi}
    \mu_k(t)
    :=
    \bE \Xi_k(t)
\end{equation}
in terms of the quantum expectation (\ref{bE}) applied to the monomials (\ref{Xi}).  The closedness means that, for any positive integer $r$ and any $r$-index $k$, the time derivative $\dot{\mu}_k$ can be expressed in terms of the mixed moments of order $r$ and lower.  This property is a corollary of  the following general result.

%%%%%%%%%%%%%%%%%%%%%%%%%%%%%%%%%%%%%%%%%%%%%%%%%%%%%%%%%%%%%%%%%%%%%%%%%%%%%%%%%%%%%%%%%%%%%%%%%%%
\begin{lem}
\label{lem:mom}
For any positive integer $r$ and any $r$-index $k := (k_1, \ldots, k_r)\in \{1, \ldots, n\}^r$,
the mixed moment $\mu_k$ from (\ref{muXi})  for the plant observables, governed by the QSDE (\ref{XQSDEgen}),  satisfies
\begin{align}
\nonumber
    \dot{\mu}_k
     =&
    \sum_{j=1}^r
    \bE (\Xi_{k_1\ldots k_{j-1}}F_{k_j}\Xi_{k_{j+1}\ldots k_r})\\
\label{mudot}
        +&\!\!\!\!
        \sum_{s,u=1}^m
        \!\!\omega_{su}\!\!\!\!\!\!
        \sum_{1\< j<\ell \< r}\!\!\!\!\!\!
        \bE
        (
        \Xi_{k_1\ldots k_{j-1}}
        g_{k_j s}
        \Xi_{k_{j+1}\ldots k_{\ell-1}}
        g_{k_{\ell}u}
        \Xi_{k_{\ell +1}\ldots k_r}
        ).\!\!\!\!
\end{align}
Here, $F_p$ is the $p$th entry of the drift vector $F$, and $g_{ps}$ denotes the $(p,s)$th entry of the dispersion matrix $G$.
%, and $\omega_{su}$ are the entries of the quantum Ito matrix $\Omega$ from (\ref{Omega}).
%,  and $\omega_{su}$ is the $(s,u)$th entry of the quantum Ito matrix $\Omega$ from (\ref{Omega}).
\end{lem}
%%%%%%%%%%%%%%%%%%%%%%%%%%%%%%%%%%%%%%%%%%%%%%%%%%%%%%%%%%%%%%%%%%%%%%%%%%%%%%%%%%%%%%%%%%%%%%%%%%%
\begin{proof}
By applying a multivariate version of the quantum Ito formula to $\Xi_k$ and using the QSDE (\ref{XQSDEgen}) together with the quantum Ito product rules \cite{P_1992} $(\rd t)^2 = 0$, $(\rd t) \rd W = 0$  and (\ref{Omega}), it follows that
\begin{align}
\nonumber
    \rd \Xi_k
     = &
    \sum_{j=1}^r
    \Xi_{k_1\ldots k_{j-1}}
    \rd X_{k_j} \\
\nonumber
    &\x
    \Big(
        \Xi_{k_{j+1}\ldots k_r}
        +
        \sum_{\ell = j+1}^r
        \Xi_{k_{j +1}\ldots k_{\ell-1}}
        (\rd X_{k_{\ell}})
        \Xi_{k_{\ell +1}\ldots k_r}
    \Big)\\
\nonumber
    = &
    \sum_{j=1}^r
    \Xi_{k_1\ldots k_{j-1}}
    \Big(
        (F_{k_j} \rd t + g_{k_j\bullet} \rd W)
        \Xi_{k_{j+1}\ldots k_r}\\
\label{dXik}
    &
        +
        \sum_{s,u=1}^m
        \omega_{su}
        g_{k_j s}
        \sum_{\ell = j+1}^r
        \Xi_{k_{j+1}\ldots k_{\ell-1}}
        g_{k_{\ell}u}
        \Xi_{k_{\ell +1}\ldots k_r}
    \Big),
\end{align}
where $g_{p\bullet}$ denotes the $p$th row of the dispersion matrix $G$.
%We have used a multivariate extension of the quantum Ito formula which, for example, in application to cubic monomials, gives $\rd (\xi \eta\zeta) = (\rd \xi)\eta \zeta + \xi(\rd \eta)\zeta + \xi\eta \rd\zeta + (\rd \xi)(\rd \eta) \zeta + (\rd \xi)\eta \rd \zeta +\xi (\rd \eta)\rd \zeta$.
The ODE (\ref{mudot}) can now be obtained by averaging both sides  of (\ref{dXik}) and using the special structure (\ref{rho0}) of the quantum state $\rho(0)$.
\end{proof}
%%%%%%%%%%%%%%%%%%%%%%%%%%%%%%%%%%%%%%%%%%%%%%%%%%%%%%%%%%%%%%%%%%%%%%%%%%%%%%%%%%%%%%%%%%%%%%%%%%%

It follows from Lemma~\ref{lem:mom} that if $F$ is an affine function of $X$ and the dispersion matrix $G$ is constant, as in the case where the plant-field coupling operators are linear in $X$, satisfying the CCRs,  then the right-hand side of the ODE (\ref{mudot}) is a linear combination of the mixed moments $\mu_{\nu}$, where the multiindices $\nu:= (\nu_1, \ldots, \nu_d)$ have dimensions  $d\< r$. More precisely,
\begin{equation}
\label{momclos}
    \dot{\mu}_k
    =
    \sum_{d=0}^r
    \sum_{\nu \in \{1, \ldots, n\}^d}
    \psi_{k,\nu} \mu_{\nu},
    \qquad
    k \in \{1, \ldots, n\}^r,
\end{equation}
where $\psi_{k,\nu}$ are complex numbers which are found from (\ref{mudot}), and the convention $\mu_{\emptyset}:= 1$ is used for the moment associated with the $0$-index.  Equivalently, an infinite dimensional vector
$
    \mu
    :=
    (\mu_k)_{k\in \{1,\ldots, n\}^r,\, r\> 0}
$,
formed by $\mu_{\emptyset}$ and the mixed moments $\mu_k$ for all possible $n^r$ multiindices $k\in \{1\, \ldots, n\}^r$ of orders $r=1,2,3,\ldots$, satisfies a system of linear ODEs $\dot{\mu} = \Psi \mu$, where $\Psi:= (\psi_{k,\nu})$ is an infinite-dimensional block-lower triangular matrix. The diagonal block  of $\Psi$ associated with the moments of order $r$ is a matrix of order $n^r$.       Hence, the solution of the system of ODEs (\ref{momclos}) can be represented as $\mu(t) = \re^{\Psi t}\mu(0)$, provided all the moments of the initial plant state $\varpi(0)$ are finite. The matrix exponential $\re^{\Psi t}$ is practically computable  due to the block-lower triangular structure  of $\Psi$.  Thus, the algebraic closedness (\ref{momclos}) allows the system of linear ODEs for the moments (\ref{muXi}) to be integrated (numerically or analytically) recursively with respect to   $r$, starting from the mean values of the plant observables for $r=1$. In particular, the mean vector and the quantum covariance matrix
\begin{equation}
\label{alphaS}
    \alpha
    :=
    \bE X,
    \qquad
    S
    :=
    \cov(X)
    =
    \bE(XX^{\rT}) - \alpha\alpha^{\rT},
\end{equation}
with the latter consisting of \emph{central} moments  of second order, satisfy the ODEs
\begin{equation}
\label{alphaSdot}
    \dot{\alpha}
    =
    A \alpha +\Theta \gamma,
    \qquad
    \dot{S}
    =
    AS + S A^{\rT} + B\Omega B^{\rT},
\end{equation}
which allow  the steady-state values of the moments to be found from the appropriate algebraic equations
if
the matrix $A$, defined in (\ref{XQSDElin}), is Hurwitz.  Now, a similar reasoning shows that the moment dynamics (\ref{mudot}) retains the algebraic closedness (\ref{momclos}) in a more general case where \emph{both} the drift vector $F$ and the dispersion matrix $G$ of the QSDE (\ref{XQSDEgen}) are \emph{affine} functions of $X$ (in the linear case above, $G$ was constant). This corresponds to a wider class of open quantum systems introduced in the next section.

%Unless indicated otherwise, vectors are organised as columns. The transpose $(\cdot)^{\rT}$ applies to vectors and matrices with operator-valued entries as if the latter were scalars. In particular, the CCRs hold for self-adjoint operators which are

%%%%%%%%%%%%%%%%%%%%%%%%%%%%%%%%%%%%%%%%%%%%%%%%%%%%%%%%%%%%%%%%%%%%%%%%%%%%%%%%%%%%%%%%%%%%%%%%%%%%
%\section{QUADRATIC PLANT-FIELD COUPLING}\label{sec:quadcoup}
%%%%%%%%%%%%%%%%%%%%%%%%%%%%%%%%%%%%%%%%%%%%%%%%%%%%%%%%%%%%%%%%%%%%%%%%%%%%%%%%%%%%%%%%%%%%%%%%%%%%

%%%%%%%%%%%%%%%%%%%%%%%%%%%%%%%%%%%%%%%%%%%%%%%%%%%%%%%%%%%%%%%%%%%%%%%%%%%%%%%%%%%%%%%%%%%%%%%%%%%
\section{QUASILINEAR OPEN QUANTUM SYSTEMS}\label{sec:qlin}
%%%%%%%%%%%%%%%%%%%%%%%%%%%%%%%%%%%%%%%%%%%%%%%%%%%%%%%%%%%%%%%%%%%%%%%%%%%%%%%%%%%%%%%%%%%%%%%%%%%

Retaining the assumption of Section~\ref{sec:lincoup} that the plant observables satisfy the CCRs (\ref{Theta}), we will now consider a wider class of plant-field interactions in which the coupling operators $h_1, \ldots, h_m$ are \emph{quadratic} polynomials of the plant observables:
\begin{equation}
\label{hMF}
    h_j = (M_j  + Y_j^{\rT}/2) X,
    \qquad
    Y_j := R_j X.
\end{equation}
Here, $M_j$ denotes the $j$th row of a matrix $M\in \mR^{m\x n}$, which describes the linear part of the coupling as in (\ref{hM}), and $R_1, \ldots, R_m \in \mS_n$ are given matrices which specify the quadratic part. An equivalent vector-matrix form of (\ref{hMF}) is
\begin{equation}
\label{hMY}
    h = (M + Y^{\rT}/2) X,
    \qquad
        Y
    :=
    \begin{bmatrix}
        Y_1
        &
        \ldots
        &
        Y_m
    \end{bmatrix},
\end{equation}
where $Y$ is an $(n\x m)$-matrix with columns $Y_1, \ldots, Y_m$ whose entries are linear combinations of the plant observables.
In the case of quadratic plant-field coupling (\ref{hMF}), an argument, similar to the derivation of (\ref{HXquad}) from (\ref{Hquad}), allows the $k$th column (\ref{gk}) of the dispersion matrix $G$ from (\ref{FG}) to be computed as
\begin{equation}
\label{gkY}
    g_k
    =
    i[(M_k  + X^{\rT}R_k/2) X,\, X]
    =
    \Theta (M_k^{\rT} + Y_k),
\end{equation}
where $M_k^{\rT}$ is the $k$th column of the matrix $M^{\rT}$, which, in view of (\ref{hMY}), implies that
\begin{equation}
\label{GMY}
    G
    =
    \Theta (M^{\rT} + Y)
    =
    B
    +
    \Theta
        \begin{bmatrix}
            R_1 X & \ldots & R_m X
        \end{bmatrix},
\end{equation}
where the matrix $B$ is defined by (\ref{KB}).
Therefore, the entries of $G$ are affine functions of the plant observables. From (\ref{hMF}) and (\ref{gkY}), it follows that the contribution of the operators
\begin{align}
\nonumber
    i[h_j, g_k]
     = &
    i[(M_j  + Y_j^{\rT}/2) X,\, \Theta R_k X ]\\
\label{hjgk}
     = &
    i\Theta R_k [(M_j  \!+\! X^{\rT}R_j/2) X,\, X ]
     \!=\!
    \Theta R_k
    \Theta (M_j^{\rT} \!+\! Y_j)\!\!\!\!\!
\end{align}
to the right-hand side of (\ref{DDR}) is linear with respect to  the plant observables:
\begin{equation}
\label{lincontr}
    i
    \sum_{j,k=1}^m
    \omega_{jk}
    [h_j, g_k]
    =
    \sum_{j,k=1}^m
    \omega_{jk}
    \Theta R_k
    \Theta (M_j^{\rT} + R_j X).
\end{equation}
Thus, in the case of canonically commuting plant observables and quadratic plant-field coupling (\ref{hMF}), the dispersion matrix $G$ is an affine function of $X$, while $\cL(X)$, given by  (\ref{DDR}),  is a \emph{cubic} polynomial of $X$. The latter property suggests finding a Hamiltonian $H$ in the form of a \emph{quartic} (degree four) polynomial of the plant observables such that the corresponding cubic polynomial $i[H,X]$ counterbalances the quadratic and cubic terms in $\cL(X)$, thus making the drift vector $F$ in (\ref{FG}) an affine function of $X$:
\begin{equation}
\label{Faff}
    F = AX + \beta,
\end{equation}
where $A\in \mR^{n\x n}$ and $\beta \in \mR^n$.
 Together with $G$ depending affinely on $X$ as described by (\ref{GMY}), the resulting  quantum plant, governed by the QSDE
 \begin{align}
\nonumber
    \rd X
    & =
    (AX+\beta) \rd t
    +
    \Theta(M^{\rT} + Y) \rd W\\
 \label{qlin}
    & =
    (AX+\beta) \rd t
    +
    \Theta
    \sum_{j=1}^m
    (M_j^{\rT} + R_jX) \rd W_j,
 \end{align}
 which we will refer to as a  \emph{quasilinear} open quantum system, retains the algebraic closedness (\ref{momclos}) in the moment dynamics (\ref{mudot}) as discussed in Section~\ref{sec:mom}. In the next section, we will show that the problem of finding such a Hamiltonian $H$ is simplified significantly  by taking  physical realizability conditions into account.
\section{PRESERVATION OF CANONICAL COMMUTATION RELATIONS}\label{sec:CCR}
%%%%%%%%%%%%%%%%%%%%%%%%%%%%%%%%%%%%%%%%%%%%%%%%%%%%%%%%%%%%%%%%%%%%%%%%%%%%%%%%%%%%%%%%%%%%%%%%%%%

The commutator of two observables $\eta(t)$ and $\zeta(t)$ on the product space $\cH\ox \cF$ inherits the evolution (\ref{xit}) with the unitary matrix $U(t)$:
\begin{equation}
\label{UU}
    [\eta(t), \zeta(t)]
    =
    U(t)^{\dagger}
    [\eta(0), \zeta(0)]
    U(t).
\end{equation}
%where use has also been made of the commutator identity $[A\ox B, C\ox D] = [A,C]\ox (BD) + (CA) \ox [B,D]$.
Hence, if $\eta(0)$ and $\zeta(0)$ satisfy a CCR, that is, if $[\eta(0),\zeta(0)] = \varphi$ is the identity operator $I_{\cH\ox \cF}$ up to a complex multiplier $\varphi$, then the unitarity of $U(t)$ and (\ref{UU}) imply that  $[\eta(t), \zeta(t)] = \varphi U(t)^{\dagger}U(t) = [\eta(0),\zeta(0)]$ for all $t\> 0$. Therefore, any CCR between observables on the space $\cH\ox \cF$ is preserved in time. In particular, if the plant observables $X_1(0), \ldots, X_n(0)$ are in CCRs with each other, then, the preservation of these CCRs is a necessary condition for physical realizability (PR) of a QSDE of the form (\ref{XQSDEgen}). Here, in accordance with \cite{JNP_2008,NJP_2009} in the linear case and \cite{MP_2012} for nonlinear systems, PR is understood as existence of a plant-field energetics model, specified by the pair $(H,h)$, which generates the particular  drift vector $F$ and dispersion matrix $G$ as described by (\ref{FG}).
We will now obtain CCR preservation conditions for the quasilinear quantum plant governed by the QSDE (\ref{qlin}) which corresponds to  (\ref{XQSDEgen}) with the drift vector $F$ and dispersion matrix $G$ given by (\ref{Faff}) and (\ref{GMY}), respectively.  To this end, let $\cM$ and $\cR$ denote linear operators which map an $n$-dimensional vector $u$ and a matrix $T$ of order $n$ to two matrices of order $n$:
\begin{align}
\label{cM}
    \cM(u)
    & :=
    \Theta
    \sum_{j,k=1}^m
    J_{jk}
    (M_j^{\rT} u^{\rT} R_k + R_j u M_k )
    \Theta,\\
\label{cR}
    \cR(T)
    & :=
    \Theta
    \sum_{j,k=1}^m
    J_{jk}
    R_j T R_k
    \Theta.
\end{align}
%The antisymmetry of $\cM(u)$ and $\cR(T)$ is ensured by the antisymmetry of the matrices $J$, $\Theta$ in (\ref{CCRW}), (\ref{Theta}) and the symmetry of the matrices $R_1, \ldots, R_m$ from (\ref{hMF}).
The significance of these operators %$\cM$ and $\cR$
is clarified by the following lemma which is instrumental to the proof of Theorem~\ref{th:CCR}.

%%%%%%%%%%%%%%%%%%%%%%%%%%%%%%%%%%%%%%%%%%%%%%%%%%%%%%%%%%%%%%%%%%%%%%%%%%%%%%%%%%%%%%%%%%%%%%%%%%%
\begin{lem}
\label{lem:const}
%Suppose the plant observables satisfy the CCRs (\ref{Theta}) and the plant-field coupling operators are quadratic polynomials given by (\ref{hMF}).
%Then
The matrix $GJG^{\rT}$, associated with the dispersion matrix $G$ in (\ref{GMY}),  is a constant complex matrix (independent of $X$ and the  initial quantum state of the plant) if and only if the operators $\cM$ and $\cR$, defined by (\ref{cM}) and (\ref{cR}), both vanish on $\mR^n$ and $\mS_n$, respectively. In this case,
\begin{equation}
\label{GJGconst}
    GJG^{\rT}
    =
    BJB^{\rT}
    -
    i\cR(\Theta)/2,
\end{equation}
where the matrix $B$ is defined by (\ref{KB}).
\end{lem}
%%%%%%%%%%%%%%%%%%%%%%%%%%%%%%%%%%%%%%%%%%%%%%%%%%%%%%%%%%%%%%%%%%%%%%%%%%%%%%%%%%%%%%%%%%%%%%%%%%%

\begin{proof}
By combining (\ref{KB}), (\ref{GMY}) with  (\ref{cM}), (\ref{cR}), it follows that
\begin{align}
\nonumber
    GJG^{\rT}
    = &
    -\Theta (M^{\rT}+Y)J(M+Y^{\rT})\Theta\\
\nonumber
     = &
    BJB^{\rT}
    -
    \Theta (M^{\rT} J Y^{\rT} + YJM) \Theta
    -\Theta Y J Y^{\rT}\Theta\\
\nonumber
     = &
    BJB^{\rT} - \cM(X)
    -
    \cR(XX^{\rT})\\
\label{GJG}
     = &
    BJB^{\rT}
    -
    \cM(X)
    -
    \cR(\mho)
         - i\cR(\Theta)/2.
\end{align}
where the operator matrix $XX^{\rT}$ is split into a symmetric part $\mho$ and the antisymmetric part
$[X,X^{\rT}]/2 = i\Theta /2$ as
\begin{equation}
\label{XX}
    XX^{\rT} = \mho + i\Theta/2,
    \qquad
    \mho
    :=
    \big(
        XX^{\rT} + (XX^{\rT})^{\rT}
    \big)/2.
\end{equation}
Hence, if $\cM=0$ on $\mR^n$ and $\cR=0$ on $\mS_n$, then the terms $\cM(X)$ and $\cR(\mho)$ vanish in (\ref{GJG}), thus proving the sufficiency of these conditions, and the relation (\ref{GJGconst}) follows. In order to prove the necessity, suppose $GJG^{\rT}$ is a constant complex matrix. Then it coincides with its expectation over any quantum state where $X$ has finite second moments, that is,
\begin{equation}
\label{GJG1}
    GJG^{\rT}
    =
    BJB^{\rT}
     - i\cR(\Theta)/2
    -
    \cM(\alpha)
    -
    \cR(\alpha\alpha^{\rT}+\Sigma),
\end{equation}
where $\alpha$ is the mean vector and
\begin{equation}
\label{Sigma}
    \Sigma := \Re S
\end{equation}
is the real part of the quantum covariance matrix of $X$  from (\ref{alphaS}), and use is made of the relation $\bE \mho = \Sigma + \alpha\alpha^{\rT}$ which follows from (\ref{XX}). By using a Gaussian initial quantum state \cite{KRP_2010} for $X$ with an arbitrary mean vector $\alpha$ and the quantum covariance matrix $S= \Sigma+i\Theta/2$, where $\Sigma \in \mS_n$  is varied independently of $\alpha$ subject to  $S \succcurlyeq 0$, it follows from the constancy of the matrix $GJG^{\rT}$ in  (\ref{GJG1}) that the linear operators $\cM$ and $\cR$ vanish on the spaces $\mR^n$ and $\mS_n$, respectively. This establishes the necessity and completes the proof.
\end{proof}
%%%%%%%%%%%%%%%%%%%%%%%%%%%%%%%%%%%%%%%%%%%%%%%%%%%%%%%%%%%%%%%%%%%%%%%%%%%%%%%%%%%%%%%%%%%%%%%%%%%

Since $\Theta$ is antisymmetric, the matrix $\cR(\Theta)$ in (\ref{GJG}) does not have to vanish under the assumption of Lemma~\ref{lem:const} that $\cR=0$ on $\mS_n$. A sufficient condition for this assumption to hold can be obtained by using the vectorization of matrices \cite{M_1988}:
\begin{equation}
\label{RR}
    \sum_{j,k=1}^m J_{jk} (\Theta R_j)\ox (\Theta R_k) = 0.
\end{equation}
More precisely, the condition (\ref{RR}) is necessary and sufficient for the operator $\cR$ in (\ref{cR}) to vanish on the space $\mR^{n\x n}$ which contains $\mS_n$. For what follows, we define, in a similar fashion to (\ref{cM}), (\ref{cR}), linear operators $\cE$ and $\cV$ which map an $n$-dimensional vector $u$ and a matrix $T$ of order $n$ to two matrices of order $n$:
\begin{align}
\label{cE}
    \cE(u)
    & :=
    \Theta
    \sum_{j,k=1}^m
    v_{jk}
    (M_j^{\rT} u^{\rT} R_k + R_j u M_k )
    \Theta,\\
\label{cV}
    \cV(T)
     & :=
    \Theta
    \sum_{j,k=1}^m
    v_{jk}
    R_j T R_k
    \Theta,
\end{align}
where $v_{jk}$ are the entries of the real part of the quantum Ito  matrix from (\ref{V}).

%%%%%%%%%%%%%%%%%%%%%%%%%%%%%%%%%%%%%%%%%%%%%%%%%%%%%%%%%%%%%%%%%%%%%%%%%%%%%%%%%%%%%%%%%%%%%%%%%%%
\begin{thm}
\label{th:CCR}
The quasilinear QSDE (\ref{qlin}) preserves the CCR matrix $\Theta$ of the plant observables from (\ref{Theta}) if and only if the linear operators $\cM$ and $\cR$, defined by (\ref{cM}), (\ref{cR}), vanish on $\mR^n$ and $\mS_n$, respectively, and
\begin{equation}
\label{PR}
        A\Theta + \Theta A^{\rT}
        + BJB^{\rT}
        =
        \cV(\Theta).
\end{equation}
\end{thm}
%%%%%%%%%%%%%%%%%%%%%%%%%%%%%%%%%%%%%%%%%%%%%%%%%%%%%%%%%%%%%%%%%%%%%%%%%%%%%%%%%%%%%%%%%%%%%%%%%%%

\begin{proof}
By combining the quantum Ito formula with the bilinearity of the commutator as $\rd [ X, X^{\rT}] = [\rd X, X^{\rT}] + [X, \rd X^{\rT}]+[\rd X, \rd X^{\rT}]$, and using  (\ref{XQSDEgen}) together with the quantum Ito product rules, it follows that
\begin{align}
\nonumber
    \rd [X, X^{\rT}&]
%     =
%     [\rd  X, X^{\rT}]
%     +
%     [X, \rd  X^{\rT}]
%     +
%     [\rd X, \rd  X^{\rT}]\\
%\nonumber
%    =&
=
    [
        F\rd t + G \rd W,
        X^{\rT}
    ] +
        [
            X,
        (F\rd t + G \rd W)^{\rT}
    ]\\
\nonumber
     &+
    [G \rd W, (G \rd W)^{\rT}]\\
\nonumber
    =&
    \Big(
        [F,X^{\rT}] + [X,F^{\rT}] + iGJG^{\rT} + \sum_{j,k=1}^m \overline{\omega_{jk}} [g_j, g_k^{\rT}]
    \Big)
    \rd t\\
\nonumber
     &+
    \sum_{k=1}^{m}
    \big(
        [g_k,X^{\rT}]
        +
        [X, g_k^{\rT}]
    \big)
    \rd W_k\\
\label{dXX}
    = &
    i
    \Big(
        A\Theta + \Theta A^{\rT}
        + GJG^{\rT} - \Theta  \sum_{j,k=1}^m \overline{\omega_{jk}} R_j \Theta R_k \Theta
    \Big)    \rd t.\!\!\!\!\!\!
\end{align}
Here, use has also been made of the commutativity between $\rd W$ and the adapted processes $F$, $G$, $X$ and the relationship
\begin{align}
\nonumber
    [G\rd W, (G\rd W)^{\rT}]
      =&
    \sum_{j,k=1}^m
    [g_j\rd W_j, g_k^{\rT} \rd W_k]\\
\nonumber
      = &
    \!\!\!\sum_{j,k=1}^m\!\!\!
    \Big(
        [g_j, g_k^{\rT}] \omega_{jk}
        \!+\! i g_j g_k^{\rT}J_{jk} -i[g_j, g_k^{\rT}]J_{jk}
        \Big)\rd t\\
%\nonumber
%    =&
%    \Big(iGJG^{\rT} + \sum_{j,k=1}^m \overline{\omega_{jk}} [g_j, g_k^{\rT}]\Big)\rd t\\
\label{GGGG}
    =&
    i
    \Big(
        GJG^{\rT}
        -
        \Theta
        \sum_{j,k=1}^m
        \overline{\omega_{jk}}
        R_j \Theta R_k \Theta
    \Big)\rd t,
\end{align}
where $J_{jk}=2\Im\omega_{jk}$ is the $(j,k)$th entry of the matrix $J$ from (\ref{CCRW}), so that $\omega_{jk}-iJ_{jk} = \overline{\omega_{jk}}$.  In turn, (\ref{GGGG}) follows from the commutator identity
$
    [\varphi\psi,\sigma \tau] = [\varphi,\sigma]\psi \tau + \varphi[\psi ,\sigma]\tau
    +\sigma[\varphi,\tau]\psi  + \sigma\varphi[\psi ,\tau]
$
which reduces to
\begin{align}
\nonumber
    [\varphi \psi,\sigma\tau]
    & =
    [\varphi,\sigma]\psi\tau + \sigma\varphi[\psi,\tau]\\
\label{ABCD}
    & =
    [\varphi,\sigma]\psi\tau + \varphi\sigma [\psi,\tau] - [\varphi,\sigma][\psi,\tau],
\end{align}
provided $[\psi,\sigma] = 0$ and $[\varphi,\tau] = 0$. More precisely,
(\ref{ABCD}) is applied to the case where  $\varphi$, $\sigma$ are entries of the dispersion matrix $G$, while $\psi$, $\tau$ are those of $\rd W$. Note that the following term in (\ref{dXX})
\begin{equation}
\label{Jacobi}
    [g_k,X^{\rT}] +[X, g_k^{\rT}]
    = 0
\end{equation}
vanishes regardless of the particular form (\ref{gkY}) of the columns of the dispersion matrix $G$ in the case of quadratic plant-field coupling (\ref{hMF}). In fact, (\ref{Jacobi}) follows from the general definition (\ref{gk}) and the Jacobi identity \cite{M_1998} combined with the CCRs (\ref{Theta}) whereby $[X,[h_k,X^{\rT}]] + [[h_k,X], X^{\rT}] = [h_k, [X,X^{\rT}]] = i[h_k, \Theta]= 0$. Also, in (\ref{GGGG}), we have used the relation $[g_j, g_k^{\rT}] = -\Theta R_j [X,X^{\rT}] R_k \Theta=-i\Theta R_j \Theta R_k \Theta$ which follows from (\ref{gkY}). Now, the CCRs (\ref{Theta}) are preserved in time, that is,  the left-hand side of the QSDE (\ref{dXX}) vanishes identically, if and only if so does the right-hand side:
\begin{equation}
\label{AAGJG}
        A\Theta + \Theta A^{\rT}
        + GJG^{\rT} - \Theta  \sum_{j,k=1}^m \overline{\omega_{jk}} R_j \Theta R_k \Theta = 0.
\end{equation}
The fulfillment of (\ref{AAGJG}) is only possible if $GJG^{\rT}$ is a constant complex matrix. By Lemma~\ref{lem:const}, this property is equivalent to that $\cM =0$ on $\mR^n$  and $\cR=0$ on $\mS_n$, in which case $GJG^{\rT}$ is given by (\ref{GJGconst}). Now, since
$
    \Theta
    \sum_{j,k=1}^m
    \overline{\omega_{jk}}
    R_j \Theta R_k
    \Theta
    =
    \cV(\Theta)
    -
    i\cR(\Theta)    /2
$
in view of $\overline{\omega_{jk}} = v_{jk}-iJ_{jk}/2$ and (\ref{cR}), (\ref{cV}),
then substitution of (\ref{GJGconst}) into (\ref{AAGJG}) yields
$
    A\Theta + \Theta A^{\rT}
    +BJB^{\rT}
    -
    \cV(\Theta)=0,
$
which is equivalent to (\ref{PR}), and the proof is complete. \end{proof}
%%%%%%%%%%%%%%%%%%%%%%%%%%%%%%%%%%%%%%%%%%%%%%%%%%%%%%%%%%%%%%%%%%%%%%%%%%%%%%%%%%%%%%%%%%%%%%%%%%%

Therefore, Theorem~\ref{th:CCR} imposes constraints (in terms of the operators $\cM$ and $\cR$) which the quadratic plant-field coupling operators (\ref{hMF}) have to satisfy in order to make an affine drift term of the QSDE (\ref{qlin}) achievable through an appropriate choice of the plant Hamiltonian $H$.
 %which can be reduced to the class of cubic, rather than quartic, polynomials due to the following result.

%%%%%%%%%%%%%%%%%%%%%%%%%%%%%%%%%%%%%%%%%%%%%%%%%%%%%%%%%%%%%%%%%%%%%%%%%%%%%%%%%%%%%%%%%%%%%%%%%%%
\begin{lem}
\label{lem:Lquad}
Suppose the CCR matrix $\Theta$ in (\ref{Theta}) is nonsingular, and the operators $\cM$ and $\cR$ in (\ref{cM}), (\ref{cR}), associated  with the quadratic plant-field coupling model (\ref{hMF}), satisfy the conditions of Theorem~\ref{th:CCR}. Then the GKSL vector $\cL(X)$ in (\ref{DDR}) is  a \emph{quadratic} polynomial of the plant observables with the leading term $\Theta YJM X/4$, that is,
\begin{equation}
\label{Lquad}
    \cL(X)
    =
    \Theta YJM X/4 + ({\rm affine\ function\ of}\ X).
\end{equation}
\end{lem}
%%%%%%%%%%%%%%%%%%%%%%%%%%%%%%%%%%%%%%%%%%%%%%%%%%%%%%%%%%%%%%%%%%%%%%%%%%%%%%%%%%%%%%%%%%%%%%%%%%%
\begin{proof} Since $\det\Theta \ne 0$, the  relation
$
    h = (M\Theta - G^{\rT})\Theta^{-1}X/2
$
between the vector $h$ of quadratic plant-field coupling  operators (\ref{hMY}) and the corresponding dispersion matrix $G$ in (\ref{GMY}) implies that
\begin{align}
\nonumber
    GJh
    & =
    GJ(M\Theta - G^{\rT})\Theta^{-1}X/2\\
%\nonumber
%    & =
%    (GJM\Theta  - GJG^{\rT})\Theta^{-1}X/2\\
\label{GJh}
    & =
    (\Theta M^{\rT}JM - GJG^{\rT}\Theta^{-1} + \Theta YJM)X/2,
\end{align}
where use is also made of
$
    GJM
 =
    \Theta (M^{\rT} + Y) JM
$. Now, if the operators $\cM$ and $\cR$ vanish on $\mR^n$ and $\mS_n$, then, by Lemma~\ref{lem:const}, the matrix $GJG^{\rT}$ is constant. Note that for an arbitrary quadratic plant-field coupling model without PR constraints,  $GJG^{\rT}$ would be quadratic and $GJh$ in (\ref{GJh}) would be a cubic polynomial.  Therefore, the constancy of $GJG^{\rT}$ reduces $GJh$ to a quadratic polynomial of $X$, with its leading (quadratic) term being  $\Theta YJM X/2$ in view of the linear dependence of $Y$ on $X$. It now remains to substitute (\ref{GJh}) and (\ref{lincontr}) into (\ref{DDR})  in order to verify that $\cL(X)$ is a quadratic polynomial of $X$  with the leading term $\Theta YJM X/4$ as in (\ref{Lquad}), where the calculation of the affine part is omitted for the sake of brevity. \end{proof}
%%%%%%%%%%%%%%%%%%%%%%%%%%%%%%%%%%%%%%%%%%%%%%%%%%%%%%%%%%%%%%%%%%%%%%%%%%%%%%%%%%%%%%%%%%%%%%%%%%%

Lemma~\ref{lem:Lquad} suggests that the class of candidate plant polynomials $H$ for counterbalancing the nonlinear terms of the GKSL operator $\cL(X)$ by $i[H,X]$ (to achieve an affine drift vector in the governing QSDE) can be reduced to \emph{cubic}  polynomials. One of such Hamiltonians is found in the next section.

%\begin{equation}
%\label{DDRquad}
%    \cL(X)
%    =
%    \frac{1}{2}
%    \Big(
%        GJh + i    \sum_{j,k=1}^m
%    \omega_{jk}
%    [h_j,g_k]
%    \Big).
%\end{equation}

%%%%%%%%%%%%%%%%%%%%%%%%%%%%%%%%%%%%%%%%%%%%%%%%%%%%%%%%%%%%%%%%%%%%%%%%%%%%%%%%%%%%%%%%%%%%%%%%%%%

%This is part of physical realizability (PR) \cite{JNP_2008}, which, in the case of linear systems, describes dynamic equivalence to an open quantum harmonic oscillator.   Here, we extend the PR conditions, known previously for linear systems, to the quasilinear systems which form a wider class; see also \cite{MP_2012} where  PR is studied for a different class of nonlinear quantum systems. The term

%%%%%%%%%%%%%%%%%%%%%%%%%%%%%%%%%%%%%%%%%%%%%%%%%%%%%%%%%%%%%%%%%%%%%%%%%%%%%%%%%%%%%%%%%%%%%%%%%%%
\section{CUBIC PLANT HAMILTONIAN}\label{sec:ham3}
%%%%%%%%%%%%%%%%%%%%%%%%%%%%%%%%%%%%%%%%%%%%%%%%%%%%%%%%%%%%%%%%%%%%%%%%%%%%%%%%%%%%%%%%%%%%%%%%%%%

The following theorem provides a characterization of the class of quasilinear quantum stochastic plants described by the QSDE (\ref{qlin}) associated with the quadratic plant-field coupling model (\ref{hMF}). For its formulation, we introduce a Hamiltonian
\begin{equation}
\label{Hcub}
    H
    :=
    \gamma^{\rT} X
    +
    X^{\rT} R_0 X/2
    -
    X^{\rT} YJMX/12,
\end{equation}
which is a cubic polynomial of the plant observables. Here, $\gamma \in \mR^n$ and $R_0 \in \mS_n$ are arbitrary vector and matrix which specify the quadratic part of $H$, while
\begin{equation}
\label{Delta}
    \Delta
     :=
    X^{\rT} YJMX
    =
        \begin{bmatrix}
        X^{\rT} R_1 X & \ldots & X^{\rT} R_m X
    \end{bmatrix}
    JM X
%    \\
%\label{Delta}
%    & =
%    \sum_{p=1}^m\sum_{s=1}^n
%    (JM)_{ps}
%    X^{\rT} R_p X X_s
\end{equation}
is a homogeneous cubic polynomial of $X$, specified by the parameters $M\in \mR^{m\x n}$ and $R_1, \ldots, R_m\in \mS_n$ of (\ref{hMF})--(\ref{hMY}).%, with $(JM)_{ps}$ the $(p,s)$th entry of the matrix $JM$.

%%%%%%%%%%%%%%%%%%%%%%%%%%%%%%%%%%%%%%%%%%%%%%%%%%%%%%%%%%%%%%%%%%%%%%%%%%%%%%%%%%%%%%%%%%%%%%%%%%%
\begin{thm}
\label{th:cub}
Suppose the observables of the open quantum plant under consideration have a nonsingular CCR matrix $\Theta $ in (\ref{Theta}), and the plant-field coupling is described by the quadratic model (\ref{hMF}) whose parameters satisfy the conditions of Theorem~\ref{th:CCR}. Then the cubic plant Hamiltonian $H$, described by (\ref{Hcub}), is self-adjoint and leads to a quasilinear QSDE (\ref{qlin}).
\end{thm}
%%%%%%%%%%%%%%%%%%%%%%%%%%%%%%%%%%%%%%%%%%%%%%%%%%%%%%%%%%%%%%%%%%%%%%%%%%%%%%%%%%%%%%%%%%%%%%%%%%%
\begin{proof}
Since the quadratic part of $H$, which was discussed in Section~\ref{sec:lincoup}, is  a self-adjoint operator which contributes an affine function of $X$ to the drift vector $F$ in (\ref{FG}), we will consider the cubic part of (\ref{Hcub}). In order to verify that $\Delta$ in (\ref{Delta}) is indeed a self-adjoint operator,  note that the definition of $Y$ in (\ref{hMF})--(\ref{hMY}) implies that the entries of the matrix $YJM$ are linear combinations of the plant observables
\begin{equation}
\label{YJM}
    YJM
    =
    \sum_{k=1}^n
    \Phi_k X_k,
\end{equation}
whose coefficients comprise matrices $\Phi_1, \ldots, \Phi_n\in \mR^{n\x n}$ as
\begin{equation}
\label{Phi}
    \Phi_k
    =
    \sum_{p,s=1}^m
    J_{ps}
    (R_p)_{\bullet k} M_s.
\end{equation}
Here, $(R_p)_{\bullet k}$ denotes the $k$th column of $R_p$, and $M_s$ is the $s$th row of $M$ as before. From the definition (\ref{cM}) of the operator $\cM$, it follows that, if $\det \Theta \ne 0$, then the condition of Theorem~\ref{th:CCR} that $\cM$ vanishes on $\mR^n$   is equivalent to the symmetry of all the matrices $\Phi_k$ in (\ref{Phi}). Hence, substitution  of (\ref{YJM}) into (\ref{Delta}) yields
$
    \Delta^{\dagger}
    =
    \big(
    \sum_{j,k,\ell=1}^n
    (\Phi_k)_{j\ell} X_jX_kX_{\ell}
    \big)^{\dagger}
    =
    \sum_{j,k,\ell=1}^n
    (\Phi_k)_{j\ell} X_{\ell}X_kX_j
    =\Delta
$,
where $(\Phi_k)_{j\ell}=(\Phi_k)_{\ell j}$ denotes  the $(j,\ell)$th entry of $\Phi_k\in \mS_n$. We will now compute the contribution of the cubic term $\Delta$ of $H$ from (\ref{Hcub}) to $i[H,X]$.  To this end,  by representing the operator $\Delta $ in (\ref{Delta}) as
\begin{equation}
\label{DZ}
    \Delta = Z^{\rT} JM X,
    \qquad
    Z:= Y^{\rT} X=(X^{\rT} R_k X)_{1\< k \< m},
\end{equation}
where the $m$-dimensional vector $Z$ consists of self-adjoint operators, it follows that
\begin{align}
\nonumber
    i[\Delta, X]
    & =
    i(Z^{\rT} JM [X,X^{\rT}] )^{\rT}
    -
    i[X,Z^{\rT}] JMX\\
\nonumber
    & =
    -(Z^{\rT} JM \Theta )^{\rT}
    +
    2\Theta Y JMX\\
\label{iD}
    & =
    -
    \Theta M^{\rT} J Y^{\rT}X
    +
    2\Theta Y JMX = 3\Theta YJM X.
\end{align}
Here, use is made of the CCRs (\ref{Theta}), the relation $-i[X,Z^{\rT}] = 2\Theta Y$ is obtained from (\ref{DZ}) by regarding the entries of $Z$ as quadratic Hamiltonians, and the symmetry of the matrix $YJM = -M^{\rT}  J Y^{\rT}$ follows from (\ref{YJM}). Now, (\ref{Hcub}) and (\ref{iD}) imply that the contribution of $\Delta$ to $i[H,X]$ is described by
$
    -i[\Delta,X]/12
    =
    -\Theta YJM X/4
$,
which is the negative of the leading quadratic term of the GKSL operator $\cL(X)$ computed in (\ref{Lquad}) of Lemma~\ref{lem:Lquad}. Therefore, $-i[\Delta,X]/12 + \cL(X)$ is an affine function of $X$ and so is the drift operator $F$ in (\ref{FG}) which corresponds to the cubic Hamiltonian (\ref{Hcub}). \end{proof}
%%%%%%%%%%%%%%%%%%%%%%%%%%%%%%%%%%%%%%%%%%%%%%%%%%%%%%%%%%%%%%%%%%%%%%%%%%%%%%%%%%%%%%%%%%%%%%%%%%%

Thus, Theorem~\ref{th:cub} constructs a physically realizable quasilinear quantum stochastic plant from an appropriately constrained quadratic plant-field coupling model and the corresponding cubic plant Hamiltonian with an arbitrary quadratic part. The matrix $A$ in (\ref{qlin}), whose calculation is omitted for the sake of brevity, can be recovered from the proofs of Lemma~\ref{lem:Lquad} and Theorem~\ref{th:cub}.

%%%%%%%%%%%%%%%%%%%%%%%%%%%%%%%%%%%%%%%%%%%%%%%%%%%%%%%%%%%%%%%%%%%%%%%%%%%%%%%%%%%%%%%%%%%%%%%%%%%
\section{QUADRATIC STOCHASTIC STABILITY}\label{sec:quadlyap}
%%%%%%%%%%%%%%%%%%%%%%%%%%%%%%%%%%%%%%%%%%%%%%%%%%%%%%%%%%%%%%%%%%%%%%%%%%%%%%%%%%%%%%%%%%%%%%%%%%%

Due to the algebraic closedness in the moment dynamics (\ref{momclos}), which extends from the linear case of Section~\ref{sec:lincoup} to the quasilinear quantum systems (\ref{qlin}), the stochastic   stability  of such systems is amenable to analysis at the level of moments of arbitrarily high  order. We will discuss the quadratic stability which is concerned with the first two moments.

%%%%%%%%%%%%%%%%%%%%%%%%%%%%%%%%%%%%%%%%%%%%%%%%%%%%%%%%%%%%%%%%%%%%%%%%%%%%%%%%%%%%%%%%%%%%%%%%%%%
\begin{thm}
\label{th:12}
For the quasilinear quantum stochastic plant governed by the QSDE (\ref{qlin}) and satisfying the CCR preservation conditions of Theorem~\ref{th:CCR}, the mean vector $\alpha$ from (\ref{alphaS}), and the real part $\Sigma$ of the quantum covariance matrix in (\ref{Sigma}) satisfy the ODEs
\begin{align}
\label{alphadot}
    \dot{\alpha}
     & =
    A\alpha + \beta,\\
\nonumber
    \dot{\Sigma}
     & =
    A\Sigma
    +
    \Sigma A^{\rT}
    -
    \cV(\Sigma)+
    BVB^{\rT}
    +
    \cR(\Theta)/4\\
\label{Sigmadot}
    &
    \quad -
    \cE(\alpha)
    -
    \cV(\alpha\alpha^{\rT}),
\end{align}
where the linear operators $\cE$ and $\cV$ are defined by (\ref{cE}), (\ref{cV}).
\end{thm}
%%%%%%%%%%%%%%%%%%%%%%%%%%%%%%%%%%%%%%%%%%%%%%%%%%%%%%%%%%%%%%%%%%%%%%%%%%%%%%%%%%%%%%%%%%%%%%%%%%%
\begin{proof}
Both (\ref{alphadot}) and (\ref{Sigmadot}) can be obtained by specializing the general moment dynamics (\ref{mudot}) to the quasilinear case of affine $F$ and $G$. Alternatively, (\ref{alphadot}) is established by averaging both sides  of (\ref{qlin}) and using the special structure (\ref{rho0}) of the quantum state $\rho(0)$. In a similar vein, by averaging the quantum Ito differential $\rd (XX^{\rT})$, it follows that the matrix
\begin{equation}
\label{Pi}
    \Pi(t)
    :=
    \bE(X(t)X(t)^{\rT})
\end{equation}
of second moments of the plant observables satisfies the ODE
\begin{equation}
\label{Pidot}
    \dot{\Pi}
    =
    A \Pi + \Pi A^{\rT} + \beta\alpha^{\rT} + \alpha\beta^{\rT}
    +
    \bE (G\Omega G^{\rT}).
\end{equation}
The representation $\Omega = V  + iJ/2$ of the quantum Ito matrix $\Omega$ from (\ref{Omega}) implies that
\begin{align}
\nonumber
    G\Omega G^{\rT}
     = &
    GVG^{\rT} + iGJG^{\rT}/2
    =
    BVB^{\rT} - \cE(X) - \cV(XX^{\rT})\\
\label{GOG}
     & + i\big(BJB^{\rT} - i\cR(\Theta)/2\big)/2,
\end{align}
where we have also used (\ref{cE}), (\ref{cV}), Lemma~\ref{lem:const} and Theorem~\ref{th:CCR}. The averaging of (\ref{GOG}) yields
$
    \bE(G\Omega G^{\rT}) = BVB^{\rT} + \cR(\Theta)/4 -\cE(\alpha) - \cV(\Pi) + iBJB^{\rT}/2
$ whose substitution into (\ref{Pidot}) leads to
\begin{align}
\nonumber
    \dot{\Pi}
     = &
    A \Pi + \Pi A^{\rT} + \beta\alpha^{\rT} + \alpha\beta^{\rT} +BVB^{\rT} + \cR(\Theta)/4\\
\label{Pidot1}
     &
     -\cE(\alpha) - \cV(\Pi) + iBJB^{\rT}/2.
\end{align}
Since the matrix $\Pi$ from (\ref{Pi}) is representable as $\Pi = \Sigma + \alpha \alpha^{\rT} + i\Theta/2$, then (\ref{Sigmadot}) is obtained  by taking the real parts on both sides of (\ref{Pidot1}) and combining the result with (\ref{alphadot}). % and $\dot{\Sigma} = \Re \dot{\Pi} - \dot{\alpha}\alpha^{\rT}-\alpha \dot{\alpha}^{\rT}$.
\end{proof}
%%%%%%%%%%%%%%%%%%%%%%%%%%%%%%%%%%%%%%%%%%%%%%%%%%%%%%%%%%%%%%%%%%%%%%%%%%%%%%%%%%%%%%%%%%%%%%%%%%%

Theorem~\ref{th:12} shows that, unlike the mean-covariance dynamics in the linear case (\ref{alphaSdot}), the Hurwitz property of the matrix $A$ is sufficient only for the stability of the quasilinear quantum plant (\ref{qlin}) at the level of the first order moments. In view of (\ref{Sigmadot}), such a plant is quadratically stable if the real parts of the eigenvalues of the linear operator $\Sigma \mapsto A\Sigma + \Sigma A^{\rT} - \cV(\Sigma)$, acting on the space $\mS_n$, are all negative. In this case, the steady-state values $\lim_{t\to +\infty}\alpha(t)$ and $\lim_{t\to +\infty}\Sigma(t)$ are unique solutions of the corresponding algebraic equations obtained by  equating the  right-hand sides of (\ref{alphadot}), (\ref{Sigmadot}) to zero. Also note that, for the quasilinear quantum plant,  the mean vector $\alpha$ influences the evolution of $\Sigma$ by entering the right-hand side of (\ref{Sigmadot}) in a quadratic fashion, whereas   the dynamics of the mean and covariances in the linear case (\ref{alphaSdot}) are completely decoupled.

%Since $\cV(T)^{\rT} = \cV(T^{\rT})$ for any matrix $T$, the operator $\cV$ is (anti-) symmetry preserving, so that $\cV(\mA_n)\subset \mA_n$ and $\cV(\mS_n)\subset \mS_n$. Furthermore, $\cV$ is negative with respect to the cone $\mS_n^+$ of real positive semi-definite symmetric matrices of order $n$ in the sense that $\cV(\mS_n^+)\subset -\mS_n^+$.

%Note, however, that, unlike the linear case, the quadratic stability of the quasilinear system does not imply the stability of its highr order moments.

%%%%%%%%%%%%%%%%%%%%%%%%%%%%%%%%%%%%%%%%%%%%%%%%%%%%%%%%%%%%%%%%%%%%%%%%%%%%%%%%%%%%%%%%%%%%%%%%%%%
%\section{CONCLUSION}\label{sec:conc}
%%%%%%%%%%%%%%%%%%%%%%%%%%%%%%%%%%%%%%%%%%%%%%%%%%%%%%%%%%%%%%%%%%%%%%%%%%%%%%%%%%%%%%%%%%%%%%%%%%%

\end{document}